\documentclass[10pt]{article}
\usepackage[utf8]{inputenc}

\usepackage{geometry}
\usepackage{amsmath, amssymb, amsthm}
\usepackage{graphicx}
\usepackage{bm}
\usepackage[linesnumbered,ruled]{algorithm2e}
\usepackage{fullpage}
\usepackage{enumerate}

\usepackage{comment}

\usepackage{lineno}

\graphicspath{{./figure/}}

\newtheorem{theorem}{Theorem}
\newtheorem{lemma}{Lemma}
\newtheorem{corollary}{Corollary}
\newtheorem{claim}{Claim}

\newcommand{\EDNR}{\textsc{EDNR}\xspace}
\newcommand{\suproot}{r}
\newcommand{\resis}[1]{r_{#1}}
\newcommand{\dem}[1]{d_{#1}}
\newcommand{\des}[2]{\ensuremath\mathrm{Des}_{#1}(#2)}

\newcommand{\losstotal}[1]{L(#1)}
\newcommand{\losssub}[2]{L(#1,#2)}

\newcommand{\sizeA}{m}
\newcommand{\sizeB}{n}

\newcommand{\Vdist}[1]{V_{#1}}
\newcommand{\Edist}[1]{E_{#1}}

\newcommand{\Tminmin}{\Tilde{T}}
\newcommand{\Fminmin}{\Tilde{F}}

\newcommand{\kstar}{\beta}
\newcommand{\subtree}[1]{\Tminmin_{#1}}

\newcommand{\treesize}[2]{a_{#1}^{#2}}

\newcommand{\Topt}{T^*}
\newcommand{\Fopt}{F^*}

\newcommand{\minmin}{\textsc{Min-Min}\xspace}

\newenvironment{listing}[1]{%
        \begin{list}{*}{%
                 \settowidth{\labelwidth}{#1}%
                 \setlength{\leftmargin}{\labelwidth}%
                 \advance \leftmargin by 12pt
                   \setlength{\itemsep}{0pt}%
                   \setlength{\parsep}{0pt}%
                   \setlength{\topsep}{0.5em}%
                   \setlength{\parskip}{0pt}%
}%
}{%
\end{list}}

\title{Loss Minimization for Electrical Flows over\\Spanning Trees on Grids
\thanks{
This work is partially supported by JSPS KAKENHI Grant Numbers 
JP24H00686, JP24H00690, 
JP20H05795, 
JP22H05001, 
JP24K02901, 
JP20K11670, JP23K10982, 
JP23K21646, 
Japan, and JST ERATO Grant Number JPMJER2301, Japan.
}
}

\author{
Takehiro Ito\thanks{Graduate School of Information Sciences, Tohoku University, Sendai, Japan} \and
Naonori Kakimura\thanks{Faculty of Science and Technology, Keio University, Yokohama, Japan} \and
Naoyuki Kamiyama\thanks{Institute of Mathematics for Industry, Kyushu University, Fukuoka, Japan} \and
Yusuke Kobayashi\thanks{Research Institute for Mathematical Sciences, Kyoto University, Kyoto, Japan} \and
Yoshio Okamoto\thanks{Graduate School of Informatics and Engineering, The University of Electro-Communications, Chofu, Japan} \\
takehiro@tohoku.ac.jp,
kakimura@math.keio.ac.jp,
kamiyama@imi.kyushu-u.ac.jp,\\
yusuke@kurims.kyoto-u.ac.jp,
okamotoy@uec.ac.jp
}
\date{}

\begin{document}

\maketitle

\begin{abstract}
We study the electrical distribution network reconfiguration problem, defined as follows.
We are given an undirected graph with a root vertex, demand at each non-root vertex, and resistance on each edge.
Then, we want to find a spanning tree of the graph that specifies the routing of power from the root to each vertex so that all the demands are satisfied and the energy loss is minimized.
This problem is known to be NP-hard in general. 
When restricted to grids with uniform resistance and the root located at a corner, 
Gupta, Khodabaksh, Mortagy and Nikolova~[Mathematical Programming 2022]
invented the so-called \minmin algorithm whose approximation factor is theoretically guaranteed.
Our contributions are twofold.
First, we prove that the problem is NP-hard even for grids;
this resolves the open problem posed by Gupta et al.
Second, we give a refined analysis of the \minmin algorithm and improve its approximation factor under the same setup. 
In the analysis, we formulate the problem of giving an upper bound for the approximation factor as a non-linear optimization problem that maximizes a convex function over a polytope, which is less commonly employed in the analysis of approximation algorithms than linear optimization problems.
\end{abstract}

\section{Introduction}

In this work, we study the electrical distribution network reconfiguration~(\EDNR for short) problem from a theoretical-computer-scientific viewpoint.
This problem was first introduced by Merlin and Back~\cite{MerlinBack} and has been studied mainly in the power delivery literature.
In the most elementary form, which we adopt in this paper, the \EDNR problem takes a power distribution network as an input.
A power distribution network is modeled by an undirected graph with a root vertex, demand at each non-root vertex, and resistance on each edge.
Then, we want to find a spanning tree of the graph that specifies the routing of power from the root to each vertex so that all the demands are satisfied and the energy loss is minimized.

\begin{figure}[tb]
    \centering
    \includegraphics[width=0.35\linewidth]{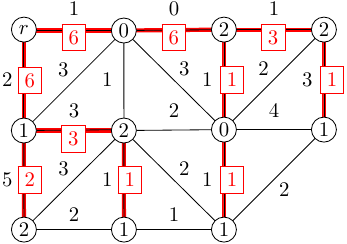}
    \caption{An instance of the \EDNR problem. The root is specified with $r$. Each non-root vertex is associated with its demand. Each edge is associated with its resistance. Red edges show a spanning tree $T$, and a red boxed number shows $\sum_{v \in \des{T}{e}}\dem{v}$ for each edge $e$ of the tree $T$.}
    \label{fig:example1}
\end{figure}

More specifically, in the \EDNR problem,
we are given a connected undirected graph $G=(V,E)$
with a fixed root $\suproot \in V$,
the demand $\dem{v} \geq 0$ for each vertex $v \in V \setminus \{\suproot \}$,
and the resistance $\resis{e} \geq 0$ for each edge $e \in E$.
Then, we want a spanning tree $T=(V,F)$ of $G$ that minimizes the following quantity, called the \emph{loss} of $T$:
\begin{equation*}
    \losstotal{T} :=
    \sum_{e \in F} \resis{e} \left(\sum_{v \in \des{T}{e}} \dem{v} \right)^2, \label{eq:obj}
\end{equation*}
where $\des{T}{e}$ is the set of vertices of the connected component~(tree) in $T-e$ that does not contain $\suproot$.
Figure~\ref{fig:example1} shows an example.

\subsection{Related work}
Most of the existing work concentrated on efficient heuristic methods to find an approximate solution~(e.g.,~\cite{CGYL193906,HCDR9423588,JKKSL1046886} and many others), but they did not give any theoretical guarantee for approximation factors.
Baran and Wu~\cite{BaranWu} gave a problem formulation in terms of combinatorial non-linear programming, but their solution method is not known to achieve optimality or any approximation guarantee.
Jabr, Singh and Pal~\cite{JSP6140618} gave a problem formulation in terms of mixed-integer convex programming and mixed-integer linear programming.
They are not polynomial-time algorithms, but, in theory, they could provide solutions with optimality guarantees. However, in their experiments, unfortunately, no optimality was guaranteed for large instances due to reasons such as time limit.
Inoue, Takano, Watanabe, Kawahara, Yoshinaka, Kishimoto, Tsuda, Minato and Hayashi~\cite{InoueTWKYKTMH14} gave an algorithm based on a binary decision diagram (BDD) to solve the \EDNR problem in which a given graph may have multiple roots. 
Their algorithm does not run in polynomial time even for a single root case, but its error bound is guaranteed theoretically.
We refer to review articles~\cite{BADRAN2017854,SULTANA2016297,THAKAR2019100191} for further background.
Note that the formulations of the \EDNR problem may vary among the literature by the choice of various realistic constraints to consider, but the essence of the problem remains the same.

Recently, more rigorous approaches have been carried out from the viewpoint of theoretical computer science.
Khodabaksh, Yang, Basu, Nikolova, Caramanis, Lianeas and Pountourakis~\cite{DBLP:conf/hicss/KhodabakhshYBNC18} proved that the \EDNR problem is NP-hard.
They also reduced the \EDNR problem to the submodular function maximization problem under matroid constraints.
It should be noted that their reduction is not approximation-factor-preserving, and thus does not yield a constant-factor approximation algorithm.
Approximation algorithms with theoretical guarantees were provided by Gupta, Khodabaksh, Mortagy and Nikolova~\cite{DBLP:journals/mp/GuptaKMN22}.
They proved that a shortest-path tree with respect to resistance is an $n$-approximate solution, where $n$ is the number of vertices.
For grids with uniform demand, uniform resistance, and the root located at a corner, they invented the so-called \minmin algorithm, and proved that the algorithm provides a $(2+o(1))$-approximate solution in polynomial time.
The \minmin algorithm also provides an approximate solution even when the condition of uniform demand is relaxed:
when all demands are within the range $[\dem{\min}, \dem{\max}]$, the algorithm provides an $\alpha^2 (2+o(1))$-approximate solution in polynomial time, where $\alpha = \dem{\max}/\dem{\min}$.
They raised an open problem that asked the NP-hardness of the \EDNR problem for grids or planar graphs.\footnote{However, as we will mention in Section~\ref{sec:hardness}, Chiba~\cite{Chiba17} had already proved the NP-hardness of the \EDNR problem for planar graphs, which are not grids, in his Master's thesis (written in Japanese).}

\subsection{Our contributions}
In this paper, we focus on the \EDNR problem with grid topology.
A \emph{grid} is defined as the Cartesian product of two paths.
More specifically, a grid is an undirected graph isomorphic to the following graph $G=(V,E)$ for some integers $n \geq 1$ and $m \geq 1$:
\begin{align*}
    V &= \{ (i,j) \mid 0 \leq i \leq n-1,\ 0 \leq j \leq m-1\},\\
    E &= \{ \{(i,j), (i',j')\} \mid |i-i'|+|j-j'| = 1\}.
\end{align*}
In that case, $G$ is called an $n \times m$-grid, and the \emph{height} of $G$ is $n$.

Our contributions are twofold.
First, we prove that the \EDNR problem is NP-hard even for grids.
This resolves the open problem mentioned above by Gupta et al.~\cite{DBLP:journals/mp/GuptaKMN22}.
We give two reductions, one based on the partition problem and the other based on the $3$-partition problem.
The former proves the NP-hardness for grids of height three, and the latter proves the NP-hardness for grids when each demand is zero or one.


Second, we revisit the \minmin algorithm by Gupta et al.~\cite{DBLP:journals/mp/GuptaKMN22}.
The algorithm focuses on the case where the underlying graph is a grid, demands and resistances are uniform and the root is located at a corner of the grid.
We give a refined analysis for the \minmin algorithm, and prove that it indeed gives a $(9/8 + o(1))$-approximate solution for the aforementioned case.
This improves the approximation factor of $2+o(1)$ given by Gupta et al.~\cite{DBLP:journals/mp/GuptaKMN22}, and also yields the improvement of the approximation factor for the non-uniform demand case to $\alpha^2 (9/8+o(1))$, where $\alpha = \dem{\max}/\dem{\min}$.

We note that the complexity status of the \EDNR problem remains open for the case where the \minmin algorithm applies, namely, for grids with uniform demand and uniform resistance, and the root located at a corner.
We think the problem is still non-trivial even for such a restricted case. 
To gain intuition, we performed a numerical optimization with CPLEX to find the optimal loss for the $n \times n$-grids, and compared them with the loss obtained by the \minmin algorithm. 
The result is summarized in \tablename~\ref{table:opt} and \figurename~\ref{fig:grid7x7}.
CPLEX was only able to find an optimal solution up to $n \leq 8$ within an hour.\footnote{%
This is not a serious computational experiment. We are content with finding (1) the \minmin algorithm does not give an optimum when $n \geq 7$ and (2) the sequence of optimal losses is not listed in the OEIS.\@
}
We checked that the sequence of optimal losses is not listed in the Online Encyclopedia of Integer Sequences (OEIS).
This suggests that it is non-trivial to find an optimal value even for this case.

\begin{table}[t]
    \centering
    \caption{The optimal loads and the loads obtained by the \minmin algorithm for the $n\times n$-grid where demands and resistances are uniform and the root is located at a corner. Up to $n \leq 6$, the \minmin algorithm gives an optimal solution, but for larger $n$, the \minmin algorithm does not give an optimal solution.}
    \begin{tabular}{rrr}
    \hline
    \multicolumn{1}{l}{$n$} & \multicolumn{1}{l}{Optimal loss}& \multicolumn{1}{l}{Loss by \minmin}\\
    \hline
    $2$ & $6$ & $6$\\
    $3$ & $52$ & $52$\\
    $4$ & $224$ & $224$\\
    $5$ & $660$ & $660$\\
    $6$ & $1\ 570$ & $1\ 570$\\
    $7$ & $3\ 242$ & $3\ 246$\\
    $8$ & $6\ 040$ & $6\ 068$\\
    \hline
    \end{tabular}
    \label{table:opt}
\end{table}

\begin{figure}[t]
\centering
\includegraphics{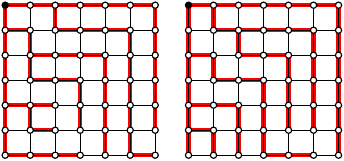}
\caption{The comparison of an optimal solution~(left) and the output of the \minmin algorithm~(right) for the $7\times 7$ grid with uniform demand and resistance. The optimal loss is $3{,}242$ while the loss from the \minmin algorithm is $3{,}246$.}
\label{fig:grid7x7}
\end{figure}

As a technical highlight, we here outline our refined analysis for the \minmin algorithm, as follows.
The \minmin algorithm constructs a spanning tree $\Tminmin$ of a given grid such that 
a similar amount of power flows on edges that are equidistant from the root in $\Tminmin$. 
We formally state and prove this property~(see Lemma~\ref{lem:balanced}) and subsequently utilize it 
to derive an upper bound on the loss $\losstotal{\tilde T}$. 
A distinctive aspect of our analysis is 
the derivation of this upper bound via a non-linear optimization problem 
that maximizes a convex function over a polytope, as formulated in \eqref{op:Bauer} in Lemma~\ref{lem:Bauer}. 
It is worth mentioning that 
convex function maximization problems are less commonly employed in the analysis of approximation algorithms, 
while there are a lot of studies using linear programs.
Additionally, evaluating the optimal value of a convex function maximization problem is generally challenging. 
Nevertheless, we successfully compute an upper bound on \eqref{op:Bauer} using convexity. 

\section{Hardness for Grids} \label{sec:hardness}

In this section, we prove that the \EDNR problem is NP-hard even for grids.
We give two incomparable NP-hardness results for grids: 
in Theorem~\ref{thm:hard_constheight}, the height of grids is fixed to three, but demands are not fixed; 
in Theorem~\ref{thm:hard_binarydemand}, demands are fixed to zero or one, but the height of grids is not fixed.  

Chiba~\cite{Chiba17} proved in his Master's thesis (written in Japanese) that the \EDNR problem is NP-hard even for planar bipartite graphs of treewidth $2$ with uniform resistance.
His reduction was based on the partition problem, although the resulting graphs were not grids.
Our reduction for Theorem~\ref{thm:hard_constheight} is also based on the partition problem and resembles his approach, but the resulting graphs in our case are grids.

We first show that the problem is NP-hard even for grids of height three. 
\begin{theorem}
\label{thm:hard_constheight}
The \EDNR problem is NP-hard on grids of height three even when $\resis{e} \in \{0,1\}$ for all edges $e$.
\end{theorem}
\begin{proof}
We reduce the following partition problem, which is well-known to be NP-complete~\cite{DBLP:books/fm/GareyJ79}.
In the problem \textsc{Partition}, we are given a positive integer $\sizeA$, and a list of $\sizeA$ positive integers $[a_1, a_2, \dots, a_{\sizeA}]$.
Then, we are asked to determine whether there exists a set $I \subseteq \{1,2,\dots,\sizeA\}$ such that
    \[
        \sum_{k \in I}a_k = \sum_{k \not\in I}a_k.
    \]

    The reduction proceeds as follows.
    We are given an instance $(\sizeA, [a_1,a_2,\dots,a_{\sizeA}])$ of \textsc{Partition}.
    Consider the $3\times (\sizeA+1)$-grid. 
    See \figurename~\ref{fig:hardness1}.
    The root $\suproot$ is located at $(0,0)$.
    The demand $\dem{i,j}$ of each non-root vertex $(i,j)$ is specified by
    \[
        \dem{i,j} = 
        \begin{cases}
            a_j & \text{if } i = 1 \text{ and } 1 \leq j \leq \sizeA,\\
            0 & \text{otherwise.}
        \end{cases}
    \]
    The resistance of an edge $e$ is specified by $1$ if it is incident to the root $\suproot$, and by $0$ otherwise.
    This completes the description of our reduction.

\begin{figure}[t]
\centering
\includegraphics{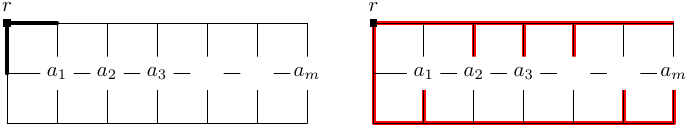}
\caption{Proof of Theorem \ref{thm:hard_constheight}. Thick edges have resistance $1$ and thin edges have resistance $0$.}
\label{fig:hardness1}
\end{figure}

    Now, we claim that the instance of \textsc{Partition} has a solution~(i.e., the answer to the question is yes) if and only if there exists a spanning tree $T$ for the constructed instance of the \EDNR problem such that 
    \[
    \losstotal{T} \leq \frac{1}{2} \left(\sum_{k = 1}^{\sizeA} a_k\right)^2.
    \]

    Suppose that there exists a set $I\subseteq \{1,2,\dots, \sizeA\}$ such that $\sum_{k\in I}a_k = \sum_{k \not\in I}a_k$.
    Then, we can construct the following spanning tree $T=(V,F)$ with
    \begin{align*}
    F &= \{\{(0,j),(0,j+1)\} \mid 0 \leq j \leq \sizeA-1\}
    \cup \{\{(2,j),(2,j+1)\} \mid 0 \leq j \leq \sizeA-1\} \\
    &\quad {} \cup \{\{(i,0), (i+1,0)\} \mid 0\leq i \leq 1\}
    \cup \{\{(0,k), (1,k)\} \mid 1\leq k \leq \sizeA, k \in I\}\\
    &\quad {} \cup \{\{(1,k), (2,k)\} \mid 1\leq k \leq \sizeA, k \not\in I\}.
    \end{align*}
    See \figurename~\ref{fig:hardness1}~(right).
    Note that the tree $T$ uses the two edges of resistance one.
    For convenience, denote 
    $e = \{(0,0), (0,1)\}$~(the horizontal edge incident to the root) and
    $f = \{(0,0), (1,0)\}$~(the vertical edge incident to the root).
    Then, we can compute the loss of $T$ as 
    \begin{align*}
        \losstotal{T}
        &= \left( \sum_{v \in \des{T}{e}} \dem{v} \right)^2 + \left( \sum_{v \in \des{T}{f}} \dem{v} \right)^2
        = \left( \sum_{k \in I} a_k \right)^2 + \left( \sum_{k \not\in I} a_k \right)^2
        = 2 \cdot \left( \frac{1}{2}\sum_{k=1}^{\sizeA} a_k \right)^2\\
        &= \frac{1}{2} \left(\sum_{k=1}^{\sizeA} a_k \right)^2.
    \end{align*}
    Thus there exists a spanning tree $T$ whose loss is bounded by $\left(\sum_{k=1}^{\sizeA} a_k \right)^2/2$.

    Next, suppose that there exists a spanning tree $T$ of loss $\losstotal{T}$ at most $\left(\sum_{k=1}^{\sizeA}a_k \right)^2/2$.
    Since the grid contains only two edges $e = \{(0,0), (0,1)\}$ and
    $f = \{(0,0), (1,0)\}$ with positive resistance, the loss is determined as
    \[
    \losstotal{T} = \left( \sum_{v \in \des{T}{e}} \dem{v} \right)^2 + \left( \sum_{v \in \des{T}{f}} \dem{v} \right)^2.
    \]
    Let $I \subseteq \{1,2,\dots,\sizeA\}$ be defined as
    $I = \{ k \in \{1,2,\dots,\sizeA\} \mid (1,k) \in \des{T}{e} \}$.
    Then, we have
    \[
    \losstotal{T} =
    \left(\sum_{k \in I} a_k\right)^2 + 
    \left(\sum_{k \not\in I} a_k\right)^2.
    \]
    By our assumption that $\losstotal{T} \leq \left(\sum_{k=1}^{\sizeA}a_k \right)^2/2$,
    it follows that
    \[
    \left(\sum_{k \in I} a_k\right)^2 + 
    \left(\sum_{k \not\in I} a_k\right)^2
    \leq 
    \frac{1}{2}\left(\sum_{k=1}^{\sizeA}a_k\right)^2
    =
    \frac{1}{2}\left(\sum_{k \in I}a_k + \sum_{k \not\in I} a_k\right)^2
    ,
    \]
    and therefore
    $\left( \sum_{k\in I}a_k - \sum_{k \not\in I} a_k\right)^2 \leq 0$.
    This implies that
       $\sum_{k\in I}a_k = \sum_{k \not\in I} a_k$, 
    which means that $I$ is a 
    solution of the given instance of \textsc{Partition}.
\end{proof}
\bigskip

We next show that the \EDNR problem is NP-hard for grids even when each demand is zero or one.
\begin{theorem}
\label{thm:hard_binarydemand}
The \EDNR problem on grids is NP-hard even when $\dem{v} \in \{0,1\}$ for all non-root vertices $v \neq r$ and $\resis{e} \in \{0,1,\infty\}$ for all edges $e$.
\end{theorem}

We note that $\resis{e} = \infty$ should be interpreted as $\resis{e}$ is a sufficiently large integer.
In the reduction, the value of $\resis{e}$ will be polynomially bounded.

\begin{figure}[!t]
\centering
\includegraphics[width=0.6\linewidth]{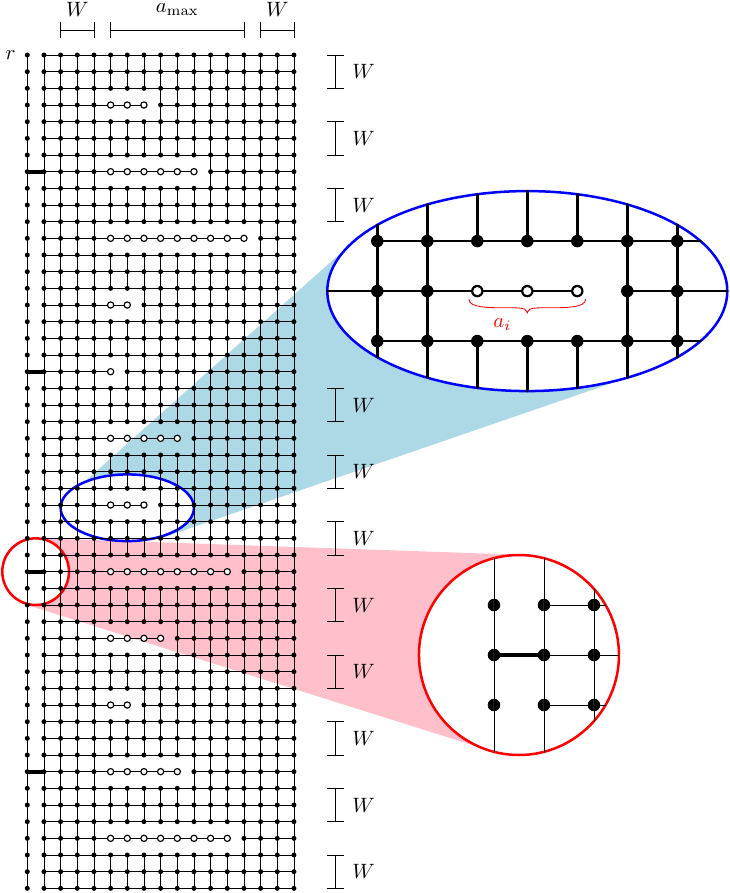}
\caption{Proof of Theorem \ref{thm:hard_binarydemand}. Thick edges have resistance $1$, thin edges have resistance $0$, and missing edges~(white edges, invisible edges) have resistance $\infty$. White vertices have demand $1$ and non-root black vertices have demand $0$. The blue window shows a close-up image of the chain with $a_i$ vertices of demand $1$. The red window shows a close-up image of one of the resistance-$1$ edges.}
\label{fig:hard2}
\end{figure}

\begin{proof}
We reduce the following $3$-partition problem, which is well-known to be strongly NP-complete~\cite{DBLP:books/fm/GareyJ79}.
In the problem \textsc{3-Partition}, we are given a positive integer $\sizeB$, and a list of $3\sizeB$ positive integers $[a_1, a_2, \dots, a_{3\sizeB}]$ with 
\[ 
    \frac{1}{4\sizeB} \sum_{k=1}^{3\sizeB}a_k < a_j < \frac{1}{2\sizeB} \sum_{k=1}^{3\sizeB} a_k 
\]
for all $j \in \{1,2,\dots,3\sizeB\}$.
Then, we are asked to determine whether there exists a partition $\{I_1, I_2, \dots, I_{\sizeB}\}$ of $\{1,2,\dots,3\sizeB\}$ into triples such that
    \[
        \sum_{k \in I_i} a_k = \frac{1}{\sizeB}\sum_{k=1}^{3\sizeB}a_k
        \quad \text{for all } i \in \{1,2,\dots,\sizeB\}. 
    \]

    The reduction proceeds as follows.
    We are given an instance $(\sizeB, [a_1,a_2,\dots,a_{3\sizeB}])$ of \textsc{3-Partition}.
    Let $a_{\max} = \max\{a_1, a_2, \dots, a_{3\sizeB}\}$ and $W$ be a sufficiently large integer that will be made more precise later.
    Consider the $(3 \sizeB (W+1)+W)\times (2+a_{\max}+2W)$-grid.
    The root $\suproot$ is located at $(0,0)$.
    See \figurename~\ref{fig:hard2} for illustration.

    Most of the non-root vertices have demand $0$.
    We now describe the vertices with demand $1$.
    For each $k \in \{1,2,\dots,3\sizeB\}$, let $i_k = k(W+1)$ and $j_k = W+2$.
    Then, we specify the demands of the vertices in the following set to $1$:
    \[
        \bigcup_{k=1}^{3\sizeB}\{(i_k,j) \mid j \in \{j_k, j_k+1, \dots, j_k + a_k-1\}\}.
    \]

    Most of the edges have resistance $0$.
    We now describe the edges with resistance $1$ and $\infty$.
    There are $\sizeB$ edges with resistance $1$.
    They are the edges of the form $\{(i,0), (i,1)\}$ for $i \in \{2(W+1), 5(W+1), \dots, (3k+2)(W+1), \dots, (3\sizeB-1)(W+1)\}$, with $k$ ranging from $0$ to $\sizeB-1$.
    There are two kinds of edges with resistance $\infty$. 
    The first kind consists of those edges of the form $\{(i,0), (i,1)\}$ that are not considered above; namely $i \not\in \{2(W+1), 5(W+1), \dots, (3k+2)(W+1), \dots, (3\sizeB-1)(W+1)\}$.
    The second kind consists of those edges that are incident to a vertex of demand $1$, with the exception that the edges of the form $\{(i,j-1), (i,j)\}$ with $i = i_k = k(W+1)$ and $j \in \{j_k, j_k+1, \dots, j_k + a_k-1\}$ have resistance $0$.
    
    This completes the description of our reduction.
    The reduction has the following properties.
    First, for each $k \in \{1,2,\dots,3\sizeB\}$, there is a chain of $a_k$ vertices with demand $1$ and that has a unique resistance-$0$ edge incident to the outside of the chain.
    Therefore, to satisfy the demands of those $a_k$ vertices, that resistance-$0$ edge must be used by an optimal spanning tree.
    Second, there are $\sizeB$ edges with resistance $1$ and any path from the root to a vertex with demand $1$ must go through one of those edges.
    Therefore, to achieve a low loss, the total demand must be divided into $\sizeB$ parts as equally as possible.
    We here omit a formal proof, because it can be obtained by similar arguments as in the proof of Theorem~\ref{thm:hard_constheight}. 
%
\begin{figure}[ht]
\centering
\includegraphics{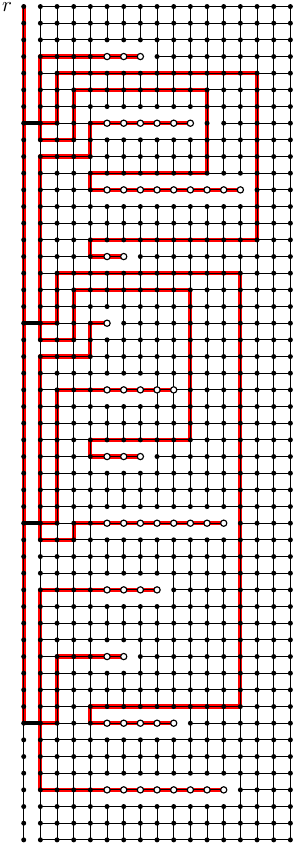}
\caption{Construction of an optimal spanning tree in the proof of Theorem \ref{thm:hard_binarydemand}, shown by red edges. In the figure, several black vertices are not incident with any red edges, but they can be connected to the existing red tree only via black edges~(i.e., edges of resistance $0$), which incurs no extra cost. }
\label{fig:hard3}
\end{figure}

    Then, there exists a spanning tree of loss at most
    $\sizeB\left(\frac{1}{\sizeB}\sum_{k=1}^{3\sizeB}a_k\right)^2$
    in the constructed instance
    if and only if the \textsc{$3$-partition} instance has a solution.
    See \figurename~\ref{fig:hard3}.
    Such a spanning tree can be constructed if $W$ is sufficiently large.
    For example, we may choose $W$ by $\Theta(\sizeB^5)$.
    The resistance $\infty$ should also be replaced with a sufficiently large integer; it is enough to use a large number in
    $\Theta\left( \left(\sum_{k=1}^{3\sizeB}a_k\right)^3\right)$
    instead of $\infty$.
    This completes the proof.
\end{proof}

\section{Improved Approximation} \label{sec:apx_minmin}

Gupta et al.~\cite[Section~5]{DBLP:journals/mp/GuptaKMN22} proposed a polynomial-time approximation algorithm, called the \textsc{Min-Min} algorithm, and obtained the approximation factor of $2+o(1)$ for the \EDNR problem on grids when all demands and resistances are uniform, and the root $\suproot$ is located at a corner of the grids. 
In this section, we give a refined analysis of the \textsc{Min-Min} algorithm under the same setup, as in the following theorem; thus our analysis provides an asymptotically better approximation factor.
\begin{theorem} \label{the:apx_uniform}
Consider the \EDNR problem on grids where $\dem{v} = 1$ for all non-root vertices $v \neq \suproot$,
$\resis{e} = 1$ for all edges $e$, and the root $\suproot$ is located at a corner of the grids.
Then, the problem admits a polynomial-time approximation algorithm within the factor of $9/8 + o(1)$.
\end{theorem}

Before proving the theorem above, we note that the same arguments as Gupta et al.~\cite[Theorem~4]
{DBLP:journals/mp/GuptaKMN22} yield the following corollary from Theorem~\ref{the:apx_uniform}. 
\begin{corollary}
    Consider grids such that demands satisfy $\dem{v} \in [\dem{\min}, \dem{\max}]$, with $\dem{\min}>0$, for all non-root vertices $v \neq \suproot$, $\resis{e} = 1$ for all edges $e$, and the root $\suproot$ is located at a corner of the grids.
    Then, the \EDNR problem on such grids admits a polynomial-time approximation algorithm within the factor of $\alpha^2 (9/8 + o(1))$, where $\alpha = \dem{\max}/\dem{\min}$.
\end{corollary}

\subsection{Lower bounds of the optimal value}
In the remainder of this section, we give a proof of Theorem~\ref{the:apx_uniform}. 
Assume that we are given an $n \times m$ grid $G=(V,E)$ such that $n \le m$, 
$\dem{v} = 1$ for all non-root vertices $v \in V \setminus \{\suproot\}$,
$\resis{e} = 1$ for all edges $e \in E$, and 
the root $\suproot$ is located at a corner of $G$;
more specifically, let $\suproot = (0,0) \in V$.
We may assume that $n \ge 2$ holds; otherwise $G$ is a path~(in particular, $G$ consists of a single vertex $\suproot$ if $m=1$) and hence the \EDNR problem on such a graph $G$ is trivial. 
Therefore, it holds that $\losstotal{T} \neq 0$ for any spanning tree $T$ of the given graph $G$. 

For each $k \in \{0,1,\ldots,n+m-2\}$, we denote by $\Vdist{k}$ the set of all vertices in $V$ that are located at distance exactly $k$ from the root $\suproot$, namely, $\Vdist{k} := \{ (i,j) \in V \mid i + j = k \}$. 
For convenience, we sometimes write $\Vdist{\ge k}$ to denote $\bigcup_{\ell = k}^{n+m-2} \Vdist{\ell}$. 
For each $k \in \{1,2,\ldots,n+m-2\}$, let $\Edist{k}$ be the set of all edges joining vertices in $\Vdist{k-1}$ and $\Vdist{k}$.
Additionally, for any spanning tree $T = (V,F)$ of $G$, we denote by $\losssub{T}{k}$ the sum of losses of edges in $\Edist{k}$, that is, 
\begin{equation} \label{eq:losssub_def}
    \losssub{T}{k} :=
    \sum_{e \in F \cap \Edist{k}} \resis{e} \left(\sum_{v \in \des{T}{e}} \dem{v} \right)^2 = 
    \sum_{e \in F \cap \Edist{k}} |\des{T}{e}|^2.
\end{equation}
Then, $\losstotal{T} = \sum_{k=1}^{n+m-2} \losssub{T}{k}$. 

We now give a lower bound on the loss of any spanning tree $T$ of $G$, as follows. 
\begin{lemma} \label{lem:lowerbound}
    For any spanning tree $T = (V,F)$ of $G$, it holds that 
    \begin{equation} \label{eq:left_lowerbound}
        \losstotal{T} > \sum_{k=1}^{m-1} \frac{|\Vdist{\ge k}|^2}{|\Vdist{k}|}.
    \end{equation}
    In particular, it holds that
    \begin{equation} \label{eq:right_lowerbound}
        \losstotal{T} > (\ln n - 2) n^2m^2.
    \end{equation}
\end{lemma}
\begin{proof}
We first estimate a lower bound on $\losssub{T}{k} = \sum_{e \in F \cap \Edist{k}} |\des{T}{e}|^2$ for each $k \in \{1,2,\ldots,n+m-2\}$.
Notice that $\sum_{e \in F \cap \Edist{k}} |\des{T}{e}| \ge |\Vdist{\ge k}|$ holds.
Since 
\[
\left(\sum_{e \in F \cap \Edist{k}} |\des{T}{e}|^2 \right) \cdot \left( \sum_{e \in F \cap \Edist{k}} 1^2 \right) 
\ge \left(\sum_{e \in F \cap \Edist{k}} |\des{T}{e}| \right)^2
\]
by the Cauchy–Schwarz inequality, we have 
\[
    \losssub{T}{k} = \sum_{e \in F \cap \Edist{k}} |\des{T}{e}|^2
    \ge  \frac{ \left(\sum_{e \in F \cap \Edist{k}} |\des{T}{e}| \right)^2}{|F \cap \Edist{k}|}
    \ge \frac{|\Vdist{\ge k}|^2}{|F \cap \Edist{k}|}
    \ge \frac{|\Vdist{\ge k}|^2}{|\Vdist{k}|},
\]
where the last inequality holds because $|F \cap \Edist{k}| \le |\Vdist{k}|$ holds for any spanning tree $T=(V,F)$.
Then, the inequality~\eqref{eq:left_lowerbound} can be derived as follows: 
\[
    \losstotal{T} = \sum_{k=1}^{n+m-2} \losssub{T}{k}
    > \sum_{k=1}^{m-1} \losssub{T}{k}
    \ge \sum_{k=1}^{m-1} \frac{|\Vdist{\ge k}|^2}{|\Vdist{k}|}.
\]

We then verify the inequality~\eqref{eq:right_lowerbound}.
Notice that $|\Vdist{\ell}| = \ell + 1$ for each $\ell \in \{0,1,\ldots,n-1\}$. 
Then, for each $k \in \{1,2,\ldots,n-1\}$, we have 
\[
    |\Vdist{\ge k}| = nm - \sum_{\ell=0}^{k-1} |\Vdist{\ell}| = nm - \frac{k(k+1)}{2}. 
\]
Therefore, since $m \ge n$, by the inequality~\eqref{eq:left_lowerbound}, we have 
\begin{align*}
    \losstotal{T} 
    &> \sum_{k=1}^{n-1} \frac{|\Vdist{\ge k}|^2}{|\Vdist{k}|} 
    = \sum_{k=1}^{n-1} \frac{1}{k+1} \left( nm - \frac{k(k+1)}{2} \right)^2
    = n^2m^2 \sum_{k=1}^{n-1} \frac{1}{k+1} - nm \sum_{k=1}^{n-1} k + \sum_{k=1}^{n-1}\frac{k^2(k+1)}{4}\\
    &> n^2m^2 (\ln n -1) - n^3m 
    \ge (\ln n -2) n^2m^2.
\end{align*}
In this way, we have verified the inequality~\eqref{eq:right_lowerbound}.
\end{proof}

\subsection{Upper bounds based on the \textsc{Min-Min} algorithm} \label{subsec:minmin}
In this subsection, we give an upper bound on a solution output by the \textsc{Min-Min} algorithm that Gupta et al.~proposed~\cite[Section~5]{DBLP:journals/mp/GuptaKMN22}. 
We do not describe the details of the \minmin algorithm in this paper, but the \minmin algorithm finds a spanning tree $\Tminmin= (V, \Fminmin)$ satisfying the following two properties~(a) and~(b).
Here, for a spanning tree $T$ and a vertex $v$, we denote by $T_v$ the subtree of $T$ induced by $v$ and all descendants of $v$, when we regard $T$ as a tree rooted at $\suproot$. 


\begin{listing}{(a)}
    \item[(a)] Let $k$ be any index in $\{n-1, n, \ldots, m-1\}$. 
               Then, for each vertex $v \in \Vdist{k}$, the subtree $\subtree{v}$ is a path, and these $|\Vdist{k}|$ paths have different lengths. 
        \medskip
    \item[(b)] Let $k$ be any index in $\{1, 2, \ldots, n-2\}$.
               Then, there exists exactly one vertex $z \in \Vdist{k}$ whose degree in $\Tminmin$ is three.
               Furthermore, $\subtree{z}$ consists of the two subtrees with the minimum number of descendants among $\subtree{v}$'s for $v\in \Vdist{k+1}$.
               The other vertices in $\Vdist{k} \setminus \{z\}$ are of degree two in $\Tminmin$. 
\end{listing}
We note that such a tree $\Tminmin$ can be found in polynomial time by finding $|V_{n-1}|$ disjoint paths in $V_{\geq n-1}$ and by merging subtrees with minimum number of descendants repeatedly from $V_{n-1}$ to $V_{0}$.
An example is given in \figurename~\ref{fig:minmin}.

\begin{figure}[tp]
\centering
\includegraphics[width=0.7\linewidth]{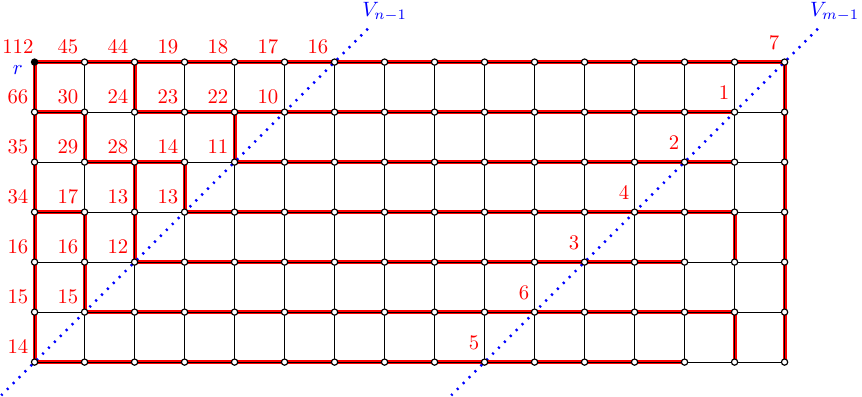}
\caption{The spanning tree $\Tminmin$ of the $7 \times 16$-grid obtained by the \minmin algorithm. The red number attached to a vertex $v$ shows the number of vertices in the subtree $\subtree{v}$.}
\label{fig:minmin}
\end{figure}

For each $k \in \{1,2,\ldots, n+m-2\}$, let $v_k^1, v_k^2, \ldots, v_k^{|\Vdist{k}|}$ be the vertices in $\Vdist{k}$, and  
we assume that these vertices are labeled so that 
$\treesize{k}{1} \le \treesize{k}{2} \le \cdots \le \treesize{k}{|\Vdist{k}|}$, where $\treesize{k}{\ell}$ denotes the number of vertices in the subtree $\subtree{v_k^\ell}$ for each $\ell \in \{1,2,\ldots,|\Vdist{k}|\}$, namely, $\treesize{k}{\ell} := |V(\subtree{v_k^\ell})|$.
Then we have the following observations from the properties~(a) and~(b).
\begin{itemize}
    \item For any index $k$ in $\{n-1, n, \dots, m-1\}$, it holds that
                \begin{align}\label{eq:treesize_path}
                    \left(\treesize{k}{1},\ \treesize{k}{2},\ \ldots,\ \treesize{k}{|\Vdist{k}|} \right) = (m-k,\ m-k+1,\ \ldots,\ n+m-k-1);
                \end{align}
               in particular, $(\treesize{m-1}{1},\ \treesize{m-1}{2},\ \ldots,\ \treesize{m-1}{|\Vdist{m-1}|}) = (1,\ 2, \dots, n)$.
        \medskip
    \item For any index $k$ in $\{1, 2, \ldots, n-2\}$, it holds that 
\begin{align}\label{eq:treesize_merge}
    \left\{\treesize{k}{1}, \treesize{k}{2},\dots, \treesize{k}{|V_k|}\right\}
    =
    \left\{\treesize{k+1}{1}+\treesize{k+1}{2}+1, \treesize{k+1}{3}+1, \treesize{k+1}{4}+1, \dots, \treesize{k}{|V_{k+1}|}+1\right\}.
    \end{align}
    If $\treesize{k}{\ell}=\treesize{k+1}{1}+\treesize{k+1}{2}+1$ for some $\ell\in \{1,2,\ldots,|\Vdist{k}|\}$, then the corresponding vertex $v_k^\ell$ is the degree-$3$ vertex of $\Vdist{k}$ in $\Tminmin$.
\end{itemize}

In what follows, we give an upper bound on the loss of $\Tminmin$ by dividing it into two parts.
For the spanning tree $\Tminmin$ of $G=(V,E)$, let $\kstar$ be the maximum integer $k \in \{1,2,\ldots,n+m-2\}$ that satisfies $\treesize{k}{|\Vdist{k}|} > 2n$.
We can assume without loss of generality that such $\kstar$ always exists; otherwise we have $|V| = nm = 1 + \treesize{1}{1} + \treesize{1}{2} \le 4n + 1$ and hence $n \le m \le 5$, which means the graph $G$ has only a constant number of vertices and we can solve the \EDNR problem optimally in constant time. 
Furthermore, since $\treesize{m-1}{|\Vdist{m-1}|} = n$ by~\eqref{eq:treesize_path}, we know that $\kstar < m-1$. 
Then it holds that 
\begin{align}\label{eq:loss_divid}
\losstotal{\Tminmin} & = \sum_{k=1}^{n+m-2} \losssub{\Tminmin}{k}
= 
\sum_{k=1}^{\kstar} \losssub{\Tminmin}{k}
+
\sum_{k=\beta+1}^{n+m-2} \losssub{\Tminmin}{k},
\end{align}
and thus we will upper-bound each term of the right-hand-side.

By the equation~\eqref{eq:losssub_def}, the total loss $\losssub{\Tminmin}{k}$ on the edges in $\Edist{k}$ can be rephrased as follows: 
for each $k \in \{1,2,\ldots, n+m-2\}$,
\begin{equation} \label{eq:lossonk}
    \losssub{\Tminmin}{k} = 
    \sum_{e \in \Fminmin \cap \Edist{k}} |\des{\Tminmin}{e}|^2 =
    \sum_{\ell=1}^{|\Vdist{k}|} (\treesize{k}{\ell})^2. 
\end{equation}
Thus $\losssub{\Tminmin}{k}$ is a quadratic function on $\treesize{k}{\ell}$'s.

We first consider the second term in~\eqref{eq:loss_divid}.
Since $\losssub{\Tminmin}{k}$ is a quadratic function on $\treesize{k}{\ell}$'s, $\losssub{\Tminmin}{k}$ can be bounded from above when $\treesize{k}{|V_k|}$ is small.
This intuition implies the following lemma. 
\begin{lemma}
It holds that 
    \begin{equation} \label{eq:right_upperbound}
        \sum_{k=\kstar+1}^{n+m-2} \losssub{\Tminmin}{k} \le 4n^2m^2.
    \end{equation}
\end{lemma}
\begin{proof}
By the choice of $\kstar$, we have 
$\treesize{k}{1} \le \treesize{k}{2} \le \cdots \le \treesize{k}{|\Vdist{k}|} \le 2n$ for all $k \in \{\kstar+1, \kstar+2, \ldots, n+m-2\}$. 
By the equation~\eqref{eq:lossonk} we have
    \[
        \sum_{k=\kstar+1}^{n+m-2} \losssub{\Tminmin}{k} 
        = \sum_{k=\kstar+1}^{n+m-2} \sum_{\ell=1}^{|\Vdist{k}|} (\treesize{k}{\ell})^2 
        \le \sum_{k=\kstar+1}^{n+m-2} \sum_{\ell=1}^{|\Vdist{k}|} (2n)^2
        \le 4n^3m,
    \]
where the last inequality holds because there are at most $nm$ vertices in the graph $G$, and hence there are at most $nm$ summands in the summation.
Therefore, since $n \le m$, the inequality~\eqref{eq:right_upperbound} holds.
\end{proof}

Therefore, it suffices to estimate the first term in~\eqref{eq:loss_divid}.
To this end, we first prove that the subtrees of $\Tminmin$ rooted at the vertices in  $\Vdist{k}$ are balanced in the following sense.
\begin{lemma} \label{lem:balanced}
    For each $k \in \{1,2,\ldots, \kstar\}$, it holds that $\treesize{k}{|\Vdist{k}|} \le 2 \treesize{k}{1}$. 
\end{lemma}
\begin{proof}
    We divide the proof into two cases, depending on whether $\subtree{v_k^{|\Vdist{k}|}}$ is a path or not. 
    
    First, consider the case where $\subtree{v_k^{|\Vdist{k}|}}$ is a path. 
    Then, since $\treesize{m-1}{|\Vdist{m-1}|} = n$ by~\eqref{eq:treesize_path}, we have
    \begin{equation} \label{eq:balance1}
        \treesize{k}{|\Vdist{k}|} \le \treesize{m-1}{|\Vdist{m-1}|} + (m-1)-k = n + m - k- 1. 
    \end{equation}
    On the other hand, as $\subtree{v_k^1}$ is a tree and $\treesize{m-1}{1}=1$ by~\eqref{eq:treesize_path}, we can estimate $\treesize{k}{1}$, as follows:
    \begin{equation} \label{eq:balance2}
        \treesize{k}{1} \ge \treesize{m-1}{1} + (m-1)-k > m - k-1.
    \end{equation}
    By the inequalities~\eqref{eq:balance1} and~\eqref{eq:balance2}, we have 
    \[
        \treesize{k}{|\Vdist{k}|} < n + \treesize{k}{1} \le 2 \cdot \max \left\{ n,\ \treesize{k}{1} \right\} = 2 \treesize{k}{1},  
    \]
    where the last equality holds because we know that $\treesize{k}{|\Vdist{k}|} > 2n$ by the choice of $\kstar$. 
    Therefore, we have proved the lemma for this case.

    Second, consider the case where $\subtree{v_k^{|\Vdist{k}|}}$ is not a path. 
    Let $z$ be the vertex in $\subtree{v_k^{|\Vdist{k}|}}$ that 
    has two children and is closest to $v_k^{|\Vdist{k}|}$~(possibly, $z = v_k^{|\Vdist{k}|}$). 
    Let $\ell$ $(> k)$ be the index with $z \in \Vdist{\ell-1}$. 
    Then, by~\eqref{eq:treesize_merge}, 
    we know that 
    \[
        |V(\subtree{z})| = |\{z\}| + \left|V\left(\subtree{v_\ell^1}\right)\right| + \left|V\left(\subtree{v_\ell^2}\right)\right| = 1+\treesize{\ell}{1}+\treesize{\ell}{2}.
    \]
    Therefore, we can estimate $\treesize{k}{|\Vdist{k}|}$ as follows:
    \begin{equation} \label{eq:balance3}
        \treesize{k}{|\Vdist{k}|} = |V(\subtree{z})|  + \ell- k - 1 
        = \treesize{\ell}{1}+\treesize{\ell}{2} + \ell - k 
        \le 2 \treesize{\ell}{2} + \ell - k. 
    \end{equation}
    On the other hand, since $\treesize{\ell-1}{1}=\min \left\{\treesize{\ell}{1}+\treesize{\ell}{2} + 1,\ \treesize{\ell}{3} + 1 \right\}$ by~\eqref{eq:treesize_merge}, we estimate $\treesize{k}{1}$ as follows:
    \begin{align} 
        \treesize{k}{1} &\ge \treesize{\ell-1}{1} + \ell - k - 1 
        = \min \left\{\treesize{\ell}{1}+\treesize{\ell}{2} + 1,\ \treesize{\ell}{3} + 1 \right\} + \ell - k -1 
        \ge \treesize{\ell}{2} + \ell - k. \label{eq:balance4}
    \end{align}
    By the inequalities~\eqref{eq:balance3} and~\eqref{eq:balance4}, we have 
    \[
        \treesize{k}{|\Vdist{k}|} \le 2 \treesize{\ell}{2} + \ell - k < 2 (\treesize{\ell}{2} + \ell - k) \le 2 \treesize{k}{1}.
    \]
    Therefore, we have proved the lemma for the second case.
\end{proof}

It follows that the sum of the squares for given real numbers can be upper-bounded when these numbers are balanced.
\begin{lemma} \label{lem:Bauer}
Let $x_1,x_2,\dots,x_t$ be non-negative real numbers such that 
$x_i \le 2x_j$ for every $i,j \in \{1,2,\dots,t\}$, and 
let $C = \sum_{i=1}^t x_i$. 
%
Then, it holds that $\sum_{i=1}^t x_i^2 \le \frac{9}{8} \cdot \frac{C^2}{t}$.
\end{lemma}
\begin{proof}
Since the lemma clearly holds when $t=1$, we assume that $t \ge 2$. 

To prove the lemma, for a fixed non-negative real number $C$, 
we will show that
the optimal value of the optimization problem below is at most $\frac{9}{8} \cdot \frac{C^2}{t}$:
\begin{equation} \label{op:Bauer}
\begin{array}{cll}
\max & \displaystyle \sum_{i=1}^tx_i^2 \\
\mbox{s.t.} & \displaystyle \sum_{i=1}^tx_i = C, \\
& x_i \le 2x_j & (\forall i,j \in \{1,2,\dots,t\}), \\
& x_i \geq 0 & (\forall i \in \{1,2,\dots,t\}). 
\end{array}
\end{equation}

\begin{figure}[ht]
    \centering
\includegraphics{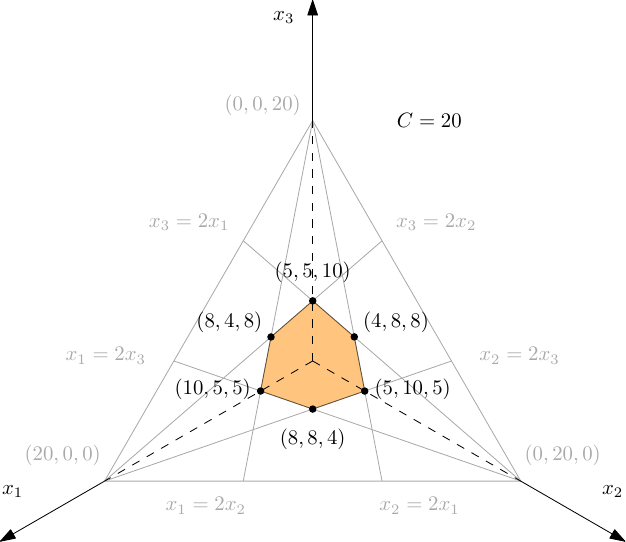}
    \caption{The feasible region of the optimization problem \eqref{op:Bauer} when $t=3$ and $C=20$. This is a $2$-dimensional polytope~(i.e., a convex polygon) with six extreme points.}
    \label{fig:feasibleregion}
\end{figure}

The feasible region of the optimization problem~\eqref{op:Bauer} forms a polytope; see Figure \ref{fig:feasibleregion}.
It is well-known~(see, e.g.,~\cite[Theorem~3]{LuenbergerY15}) that 
the problem of maximizing a convex function over a polytope has an optimal solution that is an extreme point of the polytope.
Therefore, the problem~\eqref{op:Bauer} has an optimal solution $x$ such that there do not exist two distinct feasible solutions $x', x''$ such that $x=\lambda x'+(1-\lambda)x''$ for some $\lambda\in (0, 1)$.

\begin{claim} \label{claim:Bauer}
Let $x = (x_1,x_2,\dots,x_t)$ be an optimal solution 
to the optimization problem~\eqref{op:Bauer} such that it is an extreme point.
Then, there exists a real number $\alpha$ such that 
$x_i \in \{\alpha,2\alpha\}$ for every 
$i\in\{1,2,\dots,t\}$.
\end{claim}
\noindent
\begin{proof}[Proof of Claim~\ref{claim:Bauer}]
We may assume that $C > 0$ as the claim clearly holds when $C=0$.
Let $x_{\min}$~(resp., $x_{\max}$) be the minimum~(resp., maximum) of $x_1,x_2,\dots,x_t$. 
Since $C \neq 0$ and $x_{\max} \le 2x_{\min}$, 
we have $x_{\min} > 0$. 
We note that $x_{\min}\leq x_i\leq 2x_{\min}$ for every $i \in \{1,2,\dots,t\}$.

We will prove that the claim holds with $\alpha=x_{\min}$, that is, $x_i\in\{x_{\min}, 2x_{\min}\}$ for every $i\in\{1,2,\dots,t\}$.
Suppose to the contrary that there exists $i \in \{1,2,\dots,t\}$ such that $x_{\min}<x_i<2x_{\min}$.
For simplicity, we assume that $i=t$.

Define two vectors $x^{\prime}$ and $x^{\prime\prime}$ in $\mathbb{R}^t$, respectively, as
\begin{align*}
x^{\prime}_t &:= x_t + \varepsilon, & 
x^{\prime}_i &:=  \frac{C-x_t-\varepsilon}{C - x_t}x_i\quad (i \in \{1,2,\dots,t-1\}),  \\
x^{\prime\prime}_t &:= x_t - \varepsilon, &
x^{\prime\prime}_i &:=  \frac{C-x_t+\varepsilon}{C - x_t}x_i\quad (i \in \{1,2,\dots,t-1\}), 
\end{align*}
where $\varepsilon>0$ is set to be sufficiently small so that
\begin{equation*}
\varepsilon <
\min\left\{
\frac{(C - x_t)(2x_{\min} - x_t)}{C + 2x_{\min} - x_t}, \ 
\frac{(C - x_t)(x_t - \frac{1}{2}x_{\max})}{C + \frac{1}{2}x_{\max} - x_t}\right\}.
\end{equation*}
Note that the right-hand-side of this inequality is positive, since we have
$x_t < 2x_{\min}$, 
$x_{\max} \le 2x_{\min} < 2x_t$,
and $C - x_t \ge x_{\min} > 0$. 
We observe that 
$x^{\prime}$ and $x^{\prime\prime}$ are feasible solutions to 
the problem~\eqref{op:Bauer}.
Indeed, $x^{\prime}$ is feasible to the problem~\eqref{op:Bauer}, as $x^{\prime}$
clearly satisfies the first constraint and the second constraint for every $i\in \{1,2,\dots,t-1\}$ and $j\in \{1,2,\dots,t\}$, and furthermore, since it holds that
\begin{equation*}
\varepsilon \le 
\frac{(C - x_t)(2x_{\min} - x_t)}{C + 2x_{\min} - x_t}
\ \Longleftrightarrow \ 
x_t + \varepsilon 
\le 
2 \frac{C - x_t - \varepsilon}{C - x_t}x_{\min},
\end{equation*}
we have $x_t^{\prime} \le 2 x_i^{\prime}$. 
We can prove the feasibility of $x^{\prime\prime}$ in the 
same way. 

Therefore, we have $x = \frac{1}{2}(x^{\prime} + x^{\prime\prime})$, which contradicts the fact that 
$x$ is an extreme point.
Thus, it holds that $x_i\in\{x_{\min}, 2x_{\min}\}$ for every $i\in\{1,2,\dots,t\}$, which completes the proof of the claim.
\end{proof}

Claim~\ref{claim:Bauer} implies that 
the optimization problem~\eqref{op:Bauer} can be rewritten as the following optimization problem with two variables $k$ and $\alpha$:
\begin{equation*} \label{ip:Bauer}
\begin{array}{cll}
\max & k \cdot \alpha^2 + (t - k) \cdot 4 \alpha^2\\
\mbox{s.t.} & k \cdot \alpha + (t - k) \cdot 2\alpha = C,\\
& k \in \{0,1,\dots,t\}, \alpha \geq 0. 
\end{array}
\end{equation*}
Since the first constraint means 
$\alpha = \frac{C}{2t - k}$, 
we can remove the variable $\alpha$ by substituting it: 
\begin{equation*} \label{ip2:Bauer}
\max \quad \frac{(4t - 3k)C^2}{(2t-k)^2}
\quad
\mbox{s.t.} \quad k \in \{0,1,\dots,t\}. 
\end{equation*}
Therefore, the optimal value of the optimization problem~\eqref{op:Bauer} is bounded from above by that of
\begin{equation}\label{ub:Bauer}
\max \quad \frac{(4t - 3k)C^2}{(2t-k)^2}
\quad \mbox{s.t.} \quad 0\le k\le t. 
\end{equation}
By a simple calculation, the derivative of the objective function with respect to variable $k$ is 
\[
\frac{(2t - 3k) C^2}{(2t-k)^3}. 
\]
Therefore, we can see that the maximum of~\eqref{ub:Bauer} is $\frac{9}{8} \cdot \frac{C^2}{t}$ when $k = \frac{2}{3}t$.
This completes the proof. 
\end{proof}

We now apply Lemmas~\ref{lem:balanced} and~\ref{lem:Bauer} to estimate the first term in~\eqref{eq:loss_divid}.

\begin{lemma}
It holds that
\begin{align} \label{eq:left_upperbound}
    \sum_{k=1}^{\kstar} \losssub{\Tminmin}{k} 
    & \leq 
\frac{9}{8} \sum_{k=1}^{\kstar} \frac{|\Vdist{\ge k}|^2}{|\Vdist{k}|}.
\end{align}
\end{lemma}
\begin{proof}
By Lemmas~\ref{lem:balanced} and~\ref{lem:Bauer}, it holds that 
\begin{align*}
    \sum_{k=1}^{\kstar} \losssub{\Tminmin}{k} 
    &= \sum_{k=1}^{\kstar} \sum_{\ell=1}^{|\Vdist{k}|} (\treesize{k}{\ell})^2
    \le \sum_{k=1}^{\kstar} \frac{9}{8} \left(\sum_{\ell=1}^{|\Vdist{k}|} \treesize{k}{\ell} \right)^2  \frac{1}{|\Vdist{k}|} = \frac{9}{8} \sum_{k=1}^{\kstar} \frac{|\Vdist{\ge k}|^2}{|\Vdist{k}|},
\end{align*}
where the last equality holds because the subtrees $\subtree{v_k^1}, \subtree{v_k^2}, \ldots, \subtree{v_k^{|\Vdist{k}|}}$ span all vertices in $\Vdist{\ge k}$.
\end{proof}

The following lemma completes the proof of Theorem~\ref{the:apx_uniform}.
\begin{lemma}
    Let $\Topt = (V, \Fopt)$ be a spanning tree of $G$ that is an optimal solution to the \EDNR problem on $G$. 
    Then, it holds that 
    \[
        \losstotal{\Tminmin} < \left( \frac{9}{8} + o(1) \right) \losstotal{\Topt}.
    \]
\end{lemma}
\begin{proof}
    By the upper bounds~\eqref{eq:left_upperbound} for the left part and~\eqref{eq:right_upperbound} for the right part, we have 
    \[
        \losstotal{\Tminmin} 
        = \sum_{k=1}^{\kstar} \losssub{\Tminmin}{k} + \sum_{k=\kstar+1}^{n+m-2} \losssub{\Tminmin}{k} 
        \le \frac{9}{8} \sum_{k=1}^{\kstar} \frac{|\Vdist{\ge k}|^2}{|\Vdist{k}|} + 4n^2m^2.
    \]
    Since the lower bounds~\eqref{eq:left_lowerbound} and \eqref{eq:right_lowerbound} hold also for the optimal spanning tree $\Topt$, we then have 
    \[
        \losstotal{\Tminmin} 
        \le \frac{9}{8} \losstotal{\Topt} + \frac{4}{\ln n - 2} \losstotal{\Topt} = \left( \frac{9}{8} + o(1) \right) \losstotal{\Topt}, 
    \]
    as claimed. 
\end{proof}

\section{Concluding Remark}

The most tantalizing open problem is to determine the shape of an optimal spanning tree on grids with uniform demand, uniform resistance and the root located at a corner.
We remind that the \minmin algorithm does not give an optimal solution in general.
See \figurename~\ref{fig:grid7x7}.

\bibliographystyle{acm}
\bibliography{electrical-flows}

\end{document}